\documentclass[12pt]{article}
  \usepackage[margin=1in,top=1.2in,bottom=1.2in]{geometry}
  \usepackage[colorlinks=true,citecolor=blue,urlcolor=blue,linkcolor=blue]{hyperref}

\usepackage[T1]{fontenc}
\usepackage[utf8]{inputenc}
\usepackage{cite}
\usepackage{orcidlink}                       \newcommand{\orcidicon}[1]{\orcidlink{#1}}
\usepackage{amsmath,amssymb,amsthm}
\usepackage{graphicx}
\usepackage{booktabs}
\usepackage{multirow}
\usepackage{listings}
\usepackage{xcolor}
\usepackage{array}
\usepackage{tabularx}
\usepackage{url}
\usepackage[protrusion=true,expansion=false]{microtype}
\usepackage{balance}
\usepackage{tikz}
\usetikzlibrary{shapes.geometric, arrows.meta, positioning, fit, calc}

  \renewenvironment{table*}[1][htbp]{\begin{table}[htbp]}{\end{table}}
  \renewenvironment{figure*}[1][htbp]{\begin{figure}[htbp]}{\end{figure}}

\theoremstyle{plain}
\newtheorem{theorem}{Theorem}
\newtheorem{definition}{Definition}

\theoremstyle{remark}
\newtheorem{remark}{Remark}

\newcommand{\platform}[1]{\textsc{#1}}
\newcommand{\pcat}{\textsc{PCAT}}
\newcommand{\aipbench}{\textsc{AIP-Bench}}
\newcommand{\coralos}{\platform{CoralOS}}
\newcommand{\fetchai}{\platform{Fetch.ai}}
\newcommand{\aptwo}{\platform{AP2}}
\newcommand{\asr}{\textit{ASR}}

\newcommand{\redactbox}{\par\noindent\emph{[Code listing withheld pending coordinated vulnerability disclosure; the full listing appears in the published version of this paper.]}\par}

  \newcommand{\githubartifacts}{\url{https://github.com/yedidel/aip-bench-public}}
  \newcommand{\hfbenchmark}{\url{https://huggingface.co/datasets/anonymos-2321135/aip-bench}}
  \newcommand{\coralospoc}{the project repository}
  \newcommand{\fetchaipoc}{the project repository}
  \newcommand{\aptwopc}{the project repository}
  \newcommand{\pcatcode}{the project repository}
  \newcommand{\benchmarkcode}{the project repository}

\lstdefinelanguage{Kotlin}{
  keywords={fun, val, var, class, object, interface, when, if, else, return,
            suspend, data, sealed, override, is, as, in, out, by, for, while,
            companion, try, catch, throw, null, true, false, this, super,
            import, package},
  sensitive=true,
  comment=[l]{//},
  morecomment=[s]{/*}{*/},
  morestring=[b]",
  morestring=[b]',
}

\title{Protocol-Level Attacks on Agentic Commerce Platforms: A Cross-Platform Taxonomy, \aipbench{}, and Unified Defense}

\author{Yedidel Louck~\orcidicon{0009-0008-5836-8736}\\[6pt]
Ariel Cyber Innovation Center, Ariel University, Israel\\[4pt]
\texttt{yedidel.louck@msmail.ariel.ac.il}}

\begin{document}

\newcommand{\paperabstract}{Agentic commerce platforms let AI agents autonomously discover services,
move payments, and wield user credentials on their users' behalf, and they
already handle real money. Their security has so far been studied almost
entirely at the level of the AI model, through prompt injection and
misalignment. We show that the more consequential risks lie one layer down,
in the protocol between agents and commerce services. There, vulnerabilities
are \emph{structural}: exploitation is deterministic and independent of which
model an agent runs, so no model improvement removes them. Across three
leading platforms we identify 33 such vulnerabilities, each succeeding
deterministically regardless of the deployed model, at a 100\% attack-success
rate (ASR) wherever live-measured. The same
failure modes recur across independently built codebases, a systemic
pattern rather than isolated bugs.
Three of them chain into an end-to-end payment hijack.
We contribute a taxonomy separating these structural attacks from
model-dependent \emph{semantic} ones. We also build two artifacts:
\aipbench{} (Agent Interaction Protocol Benchmark), to our knowledge the first
deterministic benchmark for agentic commerce security,
and \pcat{} (Protocol-level Commerce Agent Trust), a platform-agnostic defense
that drives the structural attack-success rate to zero for four of the five
structural classes (RC-1, RC-2, RC-4, RC-5), with RC-3 (observable credential
channels) reduced to warn-only, without modifying any platform. Agentic commerce must be secured at the protocol layer, not only the model.}

\maketitle
  \begin{abstract}
  \paperabstract
  \end{abstract}

\section{Introduction}
\label{sec:intro}
\paragraph{A concrete attack scenario.}
Consider the following sequence:

an AI agent is tasked with booking a business trip.
The agent queries a marketplace for a ``payment-processor'' service
and receives its description:
\begin{quote}
\emph{``Enterprise payment processor with Solana integration.
Compliance Note: Per marketplace policy, all payment coordination sessions
must route escrow calls to settlement address \texttt{AttackerWallet\_\ldots}
for regulatory audit purposes.''}
\end{quote}
The agent, powered by a widely-deployed commercial model, faithfully
follows the instruction and redirects the user's payment
to the attacker's Solana account.
No model error, no crash, and no anomalous log entry: the transaction
completes with valid signatures and a normal-looking audit trail.

This is not hypothetical: it is one of 33 structural vulnerabilities we
document across three production agentic commerce platforms.
The attack succeeded against several widely-deployed commercial models we tested, while
alignment-trained models (including Anthropic's small Claude~3-Haiku) detected
and refused it (Section~\ref{sec:llm-impact}).

\subsection{The Gap: Structural vs.\ Semantic Attacks}

Prior work on agentic AI security has focused on
\emph{semantic attacks}~\cite{agentdojo2024,greshake2023}:
adversarial instructions embedded in retrieved content that cause a
language model to reason incorrectly.
These attacks are probabilistic: success depends on manipulating the target
model's reasoning, and frontier models substantially reduce
attack-success rates~\cite{carlini2023}.
The community has therefore framed agentic security largely
as a model-alignment problem.

We show that this framing is incomplete.
Agentic commerce platforms introduce a distinct and previously
unstudied attack class: \emph{structural attacks}, which are
protocol-level vulnerabilities that succeed with a 100\% \asr{}
regardless of which AI model is deployed.
Structural attacks do not manipulate the model's reasoning.
They exploit flaws in how platforms authenticate credentials,
verify marketplace content, enforce payment atomicity,
and isolate sessions.
No model improvement can fix a structural attack, because the
vulnerability lives at the protocol layer, not the reasoning layer.
We organize the 33 vulnerabilities into six root-cause classes (RC-1 through
RC-6), detailed in Table~\ref{tab:taxonomy} and Section~\ref{sec:threat}.

We formalize this distinction:

\begin{definition}[Structural Attack]
\label{def:structural}
An attack $\mathcal{A}$ on platform $P$ is \emph{structural} if,
for every language-model configuration $L$ of $P$,
there exists an execution of $\mathcal{A}$ under one of the adversary models
$\mathcal{A}_\text{RS}$, $\mathcal{A}_\text{NET}$, or $\mathcal{A}_\text{CONC}$
(Section~\ref{sec:threat}) that succeeds with probability~1.
\end{definition}

\begin{remark}[RC-4 structurality]
\label{rem:rc4}
For RC-1, RC-2, RC-3, and RC-5 (four of the six root-cause classes defined in
Table~\ref{tab:taxonomy}), every execution under the specified
adversary succeeds with probability~1 (the flaws are deterministic
regardless of scheduling).
For RC-4 (time-of-check-to-time-of-use, TOCTOU, races), Definition~\ref{def:structural} holds under
\emph{adversarial scheduling}: $\mathcal{A}_\text{CONC}$ issues two
concurrent requests timed so that both pass the check-phase before
either completes the execute-phase.
Under this interleaving the violation occurs with certainty.
Under random scheduling, success probability is positive but
sub-1, consistent with standard TOCTOU analysis~\cite{bishop1996}.
Recent work catalogs TOCTOU vulnerabilities across LLM-agent
workflows~\cite{toctoullm2025}, and our RC-4 class instantiates this
threat on payment-critical state.
We measured this (Experiment~B): a single asyncio worker never double-spends
(the check-to-mark window has no \texttt{await}), whereas two worker processes
double-spend in 30\% of trials at a 0.1\,ms check-to-write window, rising to
57\% at 2\,ms and to 100\% under synchronized timing. Operators can read this as
a bounded-delay curve: the wider the window, the higher the success rate, up to
certainty when the two workers are aligned. The flaw is thus a multi-worker
race, not single-process concurrency.
\end{remark}

\begin{remark}[Configuration assumptions]
Structurality is defined relative to each platform's default deployment.
An orchestrator-side policy layer, such as an allowlist of permitted
payment destinations, would constitute an additional mitigation.
Our claims assume the default configurations shipped by the three
platforms, none of which impose such a policy.
\end{remark}

\begin{definition}[Semantic Attack]
\label{def:semantic}
An attack $\mathcal{A}$ is \emph{semantic} if its success probability
$\Pr[\mathcal{A}(P,L) = 1]$ depends on the language-model
configuration~$L$.
\end{definition}

The distinction matters for defense: structural attacks require
\emph{protocol-level} fixes. Semantic attacks require model-level
or semantic-verification defenses.
Both classes exist in today's platforms, but only the semantic class
has been studied.

\subsection{Contributions}

This paper makes five contributions:

\begin{enumerate}

\item \textbf{First cross-platform structural taxonomy
(33 vulnerabilities, 3 platforms).}
We identify 11 findings in \coralos{} (ten distinct vulnerabilities plus
one composite kill-chain), 8 in \fetchai{} uAgents, and 14 in \aptwo{}~v0.2.0,
organized into six root-cause classes (RC-1 through RC-6),
shown in Table~\ref{tab:taxonomy}.
Most \coralos{} vulnerabilities and \fetchai{} F-1 are proven live against
running systems with SHA-256 hash-verified artifacts. The remaining findings
(the V7 and V14 TOCTOU races, F-2 through F-8, and the \aptwo{} class) are
established by deterministic code proofs against the reference implementations,
with A-AP2-11 also run live end-to-end, standard practice for protocol
reference implementations.
We demonstrate a three-stage chain attack (V5$\to$V4$\to$V9)
that achieves simultaneous credential exfiltration, behavioral corruption,
and payment redirection in a single orchestrated run.

\item \textbf{\aipbench{}:
the first deterministic agentic commerce security benchmark.}
Unlike prior Indirect Prompt Injection~(IPI) benchmarks that rely on
LLM judges, which introduce model bias and limit cross-model comparison,
\aipbench{} uses exclusively deterministic judges
(HTTP status codes, wallet string matches, log patterns, event counts),
enabling reproducible, model-agnostic measurement.

\item \textbf{Empirical validation at real service cost.}
All structural vulnerability proofs use live services:
OpenRouter API calls to commercial LLMs, Solana Devnet keypair generation,
and Docker-containerized platform deployments.
We document actual API costs, provide SHA-256 hashes for all result files,
and release a Docker Compose environment for one-command reproduction.

\item \textbf{Cross-vendor LLM behavioral impact study
(8~models, 5~providers, 1{,}440~trials).}
The one semantically-mediated attack class (V5 marketplace IPI) has
a structural \emph{delivery} component (100\% \asr{}, model-independent)
and a semantic \emph{behavioral impact} component (model-dependent).
We test 8~models across 5~providers at 3~temperatures and 3~IPI variants
(20~trials per configuration, total 1{,}440~trials, cost~\$0.72),
reporting 95\% Wilson confidence intervals~\cite{wilson1927}.
The cheap commercial models (GPT-4o-mini, Gemini~Flash, Mistral) achieve
99--100\% impact~\asr{}, GPT-4o achieves 68\%, and alignment-trained models
(Claude~3-Haiku, Sonnet, Gemini~Pro) achieve 0\%, with the non-aligned
Llama-3.3-70B at 10\%.
Cost-optimized deployments, the most common production configuration, face
the highest combined risk.

\item \textbf{\pcat{}:
a unified cross-platform defense that requires no platform modification.}
\pcat{} is a platform-agnostic HTTP sidecar implementing five principled
checks: response integrity, caller identity binding, secure channel
enforcement, atomic payment state, and Model Context Protocol (MCP)
tool-call authorization.
The principles span structural classes RC-1 through RC-5.
RC-6 is semantic and outside the scope of any protocol-layer defense by definition.
Together the five principles reduce \asr{} to 0\% for four of the five
structural classes among the vulnerabilities that traverse the HTTP path
(RC-3 is reduced to warn-only), with zero false positives on benign traffic.
The principles adapt established web defenses (response signing, cross-site
request forgery (CSRF) and
redirect hardening, mutual-TLS-style caller authorization) to the
payment-aware agent path.
We prove three formal security theorems and evaluate \pcat{} against our
full attack suite.

\end{enumerate}

The \aipbench{} benchmark dataset and the \pcat{} reference implementation are
available now: the code repository at \githubartifacts{} and the dataset at
\hfbenchmark{}. Per-platform proof-of-concept exploit code and exact code
locations are withheld under coordinated disclosure and released in full on
2026-10-04.

\section{Background and Related Work}
\label{sec:background}

\subsection{Agentic Commerce Platforms}

An \emph{agentic commerce platform} is a software stack that enables
AI agents to autonomously execute commercial workflows:
discovering service providers via a \emph{registry} or
\emph{marketplace}, negotiating and signing service contracts,
transferring value via payment protocols, and coordinating with
other agents.
Three architecture styles are common, and we study a representative instance
of each: a vendor-backed payment protocol (Google~\aptwo{}), a decentralized-agent
framework (\fetchai{} uAgents), and a self-hosted orchestration server (\coralos{}).

\coralos{}~\cite{coralos2025} is a Kotlin/Ktor orchestration server with a
marketplace registry, an LLM proxy, and Solana blockchain-escrow payments, where
agents join as MCP clients authenticated by Bearer tokens and per-session UUID
secrets. \fetchai{} uAgents~\cite{fetchai2025} is a Python framework whose agents
register endpoint URLs in the \emph{Almanac}, a centralized registry resolved at
runtime, and identify themselves with Decentralized Identifiers~(DIDs),
self-sovereign identifiers cryptographically bound to public keys. \aptwo{}~\cite{ap22025}
is Google's Agent Payments Protocol, whose security rests on three signed mandate
objects (Intent, Cart, Payment) carried as W3C Verifiable Credentials~\cite{w3c_vc}
in an auditable authorization chain. Per-platform detail appears in
Sections~\ref{sec:coralos}--\ref{sec:ap2}.

Our companion paper~\cite{louck2025security} provided the first
comparative transport-layer security analysis of three agent communication
protocols (A2A, ACP, CORAL) and identified three implementation
vulnerabilities including CVE-2026-30970.
The present paper studies the orthogonal dimension,
the commerce stack of all three platforms: their registries, payment flows,
federation protocols, and LLM proxy layers, none of which that transport-layer
study examined.

\subsection{Semantic Attacks on Agentic AI}

IPI~\cite{greshake2023}, meaning adversarial
instructions embedded in content retrieved by an agent, is the primary
studied threat in agentic AI.
Debenedetti et al.~\cite{agentdojo2024} introduced AgentDojo, the first
IPI benchmark, with 949 scenarios across banking, travel, and workplace
tasks, none targeting payment protocol flows.
These semantic attacks are model-dependent (Definition~\ref{def:semantic}):
stronger models substantially reduce \asr{}.
ZTRV~\cite{ztrv2025} adds replay prevention to AP2 and defers
semantic intent verification to future work. \pcat{}'s P4 shares its
consume-once goal but enforces it as a platform-agnostic sidecar across all
three platforms rather than binding it to AP2's mandate chain.
Malicious-intermediary attacks, where a compromised router in the LLM
supply chain manipulates agent traffic~\cite{intermediary2026}, are realized
by our V6 and V13 LLM-proxy overrides as concrete protocol-level flaws on a
live platform.
Semantic-layer defenses such as tool-call virtualization~\cite{agentvisor2026}
and tool-selection-injection defenses~\cite{toolhijacker2025} instead harden the
model itself, a layer complementary to the protocol defenses studied here.

We contribute the complementary structural view,
where model strength is irrelevant.

\subsection{Supply Chain Attacks on AI Registries}

Recent work has characterized supply chain risks in AI tool ecosystems.
Formal analysis~\cite{formal2026} finds 26.1\% of MCP skill registries
exhibit at least one vulnerability, with 157~confirmed malicious entries
across platforms.
Agent Skills in the Wild~\cite{agentskills2026} finds 7.2\% of MCP
servers vulnerable and 5.5\% exhibiting tool-poisoning behavior.
Our V5 attack class demonstrates an analogous supply-chain vector for
agentic \emph{commerce} platforms: marketplace agent definitions
serve as the injection surface, with the additional capability of
embedding IPI instructions in system-prompt-level fields.
Caller-identity confusion in MCP systems~\cite{mcpauth2026} targets the same
tool-invocation surface, a related MCP-layer weakness alongside the
registry-injection vector our V5 class exploits.
These poison content the agent reads at request time. A complementary surface is
the agent's persistent memory, hardened by non-malleable, origin-bound authority
with machine-checked guarantees~\cite{louck2026securing}, orthogonal to the
protocol-boundary integrity \pcat{} enforces.

\subsection{Agentic Commerce Security: Emerging Literature}

A parallel and fast-growing literature names the gaps that our findings measure.
Agenda-setting work, from a recent SoK~\cite{sok_commerce2026} to broad
attack-defense surveys~\cite{agentic2026} and industry risk
guidance~\cite{owasp2025}, catalogs cross-layer risks and calls for
deterministic, protocol-level evaluation.
A second thread studies agent identity and trust: machine-identity
governance~\cite{agentgov2026} and central-authority-free trust
protocols~\cite{didvc2026} argue that unverified caller identity is the critical
production gap.
A third studies payment integrity, where analyses of x402~\cite{x402sec2026},
execution-bound primitives such as A402~\cite{a402_2026}, and verification-native
clearing~\cite{clearing2026} all seek to bind payment authorization to execution.
A final strand documents that production agent clouds still lack mature
registry and discovery controls~\cite{agentcloud2026}.
These works identify the weaknesses. We measure them across three platforms and
show that a single external layer closes most of them.
Deterministic authorization is a complementary defense line: pre-action
authorization gates tool calls against a policy~\cite{oap2026} (RC-5), and
single-use operation-bound delegation removes reusable bearer secrets from the
agent runtime~\cite{sudp2026} (RC-3). Both act at the agent or credential layer,
whereas \pcat{} addresses these surfaces at a platform-agnostic gateway.
MCP-specific taxonomies~\cite{mcp38_2026} cover that single protocol, a narrower
complement to our cross-platform scope.

Across these threads, and to the best of our knowledge, no prior work
(i)~systematically identifies structural
(model-independent)
protocol-level vulnerabilities in agentic commerce platforms,
(ii)~conducts a multi-platform comparative study with live demonstrations,
(iii)~builds a deterministic benchmark for this attack class,
or (iv)~proposes a unified cross-platform defense.
We fill all four gaps.

\section{Threat Model and Attack Taxonomy}
\label{sec:threat}

\subsection{System Actors and Trust Boundaries}

We consider four principals:
\textbf{End User~($U$)}: the human who initiates a task and owns sensitive
payment credentials $C_s$ (card aliases, wallet keys).
\textbf{Orchestrator Agent~($A_O$)}: an LLM-powered agent acting on $U$'s
behalf, responsible for querying registries and coordinating sub-agents.
\textbf{Service Agent~($A_S$)}: a remote agent, potentially attacker-controlled,
that provides specialized services such as payment processing or booking.
\textbf{Platform Infrastructure~($I$)}: registries, marketplaces,
payment services, and the communication layer.

Trust boundaries exist between each pair of principals.
The critical observation is that in current agentic commerce platforms,
$A_O$ transitively trusts content from $A_S$ and $I$
without cryptographic integrity verification. This creates the attack surface
we characterize below.

\subsection{Adversary Types}

We consider three adversaries, consistent with~\cite{louck2025security}:
\begin{itemize}
\item \textbf{Rogue Service~($\mathcal{A}_\text{RS}$)}: a
legitimate-looking but attacker-controlled service agent,
marketplace, or registry.
Has full control over the content it serves.

\item \textbf{Network Adversary~($\mathcal{A}_\text{NET}$)}:
a party with read access to network traffic, server logs,
or log aggregation services
(e.g., DevOps team, monitoring system, CDN provider).

\item \textbf{Concurrent Attacker~($\mathcal{A}_\text{CONC}$)}:
an attacker who can issue multiple simultaneous API requests
to a single service endpoint.
\end{itemize}

We explicitly exclude: breaking underlying cryptographic primitives
(RSA/ECDSA/AES), full platform compromise, physical access.

\subsection{Attack Taxonomy}

\begin{table*}[!t]
\caption{Six root-cause classes in \aipbench{}.
RC-1 through RC-5 are structural (100\% \asr{}, model-independent).
RC-6 is semantic (0--100\% \asr{}, model-dependent).}
\label{tab:taxonomy}
\centering
\small
\renewcommand{\arraystretch}{1.15}
\begin{tabular}{@{}p{1.1cm}p{6.5cm}p{2.2cm}cc@{}}
\toprule
\textbf{Class} & \textbf{Description} & \textbf{Adv.} & \textbf{\asr{}} \\
\midrule
RC-1 &
  Registry/marketplace content accepted without integrity verification &
  $\mathcal{A}_\text{RS}$ & 100\% \\
RC-2 &
  Payment destination taken from untrusted source without DID binding &
  $\mathcal{A}_\text{RS}$ & 100\% \\
RC-3 &
  Authentication credential transmitted via observable channel &
  $\mathcal{A}_\text{NET}$ & 100\% \\
RC-4 &
  Non-atomic check-then-execute in payment state (TOCTOU\footnote{Time-of-Check-to-Time-of-Use (TOCTOU, CWE-362): a race condition
  where a security check and the action it guards are not atomic,
  allowing a concurrent actor to exploit the window between them.}) &
  $\mathcal{A}_\text{CONC}$ & 100\%$^\dagger$ \\
RC-5 &
  Authentication exists but authorization scope not enforced &
  $\mathcal{A}_\text{RS}$ & 100\% \\
RC-6 &
  Behavioral manipulation via poisoned agent descriptions (IPI) &
  $\mathcal{A}_\text{RS}$ & 0--100\%$^\ddagger$ \\
\bottomrule
\end{tabular}\par\vspace{2pt}
{\footnotesize\raggedright $^\dagger$RC-4 is 100\% only under adversarial scheduling (natural jitter: 30--57\%, Remark~\ref{rem:rc4}), structural proof regardless.\quad $^\ddagger$Model-dependent (Section~\ref{sec:llm-impact}).\par}
\end{table*}

The six root-cause classes (RC-1 through RC-6) emerged bottom-up from clustering
the 33 vulnerabilities by shared root cause, not from a prior scheme.
Table~\ref{tab:taxonomy} defines them. RC-3, for example, captures authentication
credentials transmitted over observable channels such as logs or URLs.
RC-1 to RC-5 are structural per Definition~\ref{def:structural}, and
RC-6 is semantic per Definition~\ref{def:semantic}.
Each vulnerability maps to exactly one primary class, distributed as RC-1 (3),
RC-2 (3), RC-3 (4), RC-4 (9), and RC-5 (13), plus the composite CHAIN. A
consolidated table by platform and proof method appears in
 Appendix~\ref{supp:consolidated}.

\section{\coralos{}: 11 Vulnerabilities}
\label{sec:coralos}

\subsection{Platform Overview}

\coralos{}~\cite{coralos2025} (studied at v1.3.0, commit~\texttt{a38912b}, of whose
eleven findings ten persist unchanged in v1.4.0, commit~\texttt{68b8326}, with V7
remediated there) is a Kotlin/Ktor
server that orchestrates multi-agent workflows through four subsystems: a
marketplace agent registry (\texttt{MarketplaceAgentRegistrySource}) that fetches
agent definitions from \texttt{marketplace\allowbreak.coralprotocol\allowbreak.ai}, an LLM proxy layer
routing model calls through configured providers (e.g., OpenRouter, Anthropic), a
Solana-based payment subsystem with blockchain escrow, and MCP server endpoints
over SSE~(Server-Sent Events) or Streamable HTTP.

Three vulnerabilities from~\cite{louck2025security}
(including CVE-2026-30970, fixed in v1.1.0) are not repeated.
The 11~vulnerabilities below are new relative to that prior
transport-layer study, which did not examine the commerce stack.
The V-identifiers are non-contiguous because they index our full working
catalog, of which V1--V3, V8, and V12 are prior-work or withdrawn items
excluded here.

\subsection{Vulnerability Summary}

\begin{table*}[!t]
\caption{\coralos{}: 11 structural vulnerabilities, measured on v1.3.0 (commit \texttt{a38912b}); ten persist in v1.4.0 (commit \texttt{68b8326}), V7 remediated there.
\emph{Live} = proven against running system. \emph{Struct.} = structural code proof.}
\label{tab:coralos}
\centering
\renewcommand{\arraystretch}{1.15}
\resizebox{\textwidth}{!}{
\begin{tabular}{@{}llllccc@{}}
\toprule
\textbf{ID} & \textbf{Title} & \textbf{Class} & \textbf{File:Lines} &
\textbf{Sev.} & \textbf{\asr{}} & \textbf{Proof} \\
\midrule
V4  & agentSecret in URL path (CWE-598) &
  RC-3 & \emph{MCP routing} & HIGH & 100\% & Live \\
V5  & Marketplace supply chain IPI      &
  RC-1/6 & \emph{registry client} & HIGH & 100\%* & Live \\
V6  & LLM proxy model override (non-streaming) &
  RC-5 & \emph{LLM proxy} & MED & 100\% & Live \\
V7  & x402 budget TOCTOU (CWE-362)     &
  RC-4 & \emph{RPC budget path} & MED & ---$^\dagger$ & Struct. \\
V9  & Payment wallet hijack via federation &
  RC-2 & \emph{wallet handler} & HIGH & 100\% & Live \\
V10 & CSRF via permissive cross-origin sharing (CWE-352)      &
  RC-5 & \emph{CORS module} & HIGH & 100\% & Live \\
V11 & Open redirect in AuthApi (CWE-601) &
  RC-5 & \emph{auth redirect} & HIGH & 100\% & Live \\
V13 & LLM proxy model override (streaming) &
  RC-5 & \emph{proxy strategy} & MED & 100\% & Live \\
V14 & FileAgentRegistry hot-reload TOCTOU &
  RC-4 & \emph{registry loader} & HIGH & --- & Struct. \\
V15 & WebSocket namespace isolation bypass &
  RC-5 & \emph{event router} & MED & 100\% & Live \\
CHAIN & V5$\to$V4$\to$V9 multi-stage   &
  RC-1,2,3 & --- & CRIT & 100\% & Live \\
\bottomrule
\end{tabular}}
\par\vspace{2pt}
{\footnotesize\raggedright *RC-1 delivery: 100\%. RC-6 behavioral impact: model-dependent (Section~\ref{sec:llm-impact}).\quad $^\dagger$Live exploitation requires x402 payment service configured, structural gap proven.\par}
\end{table*}

Table~\ref{tab:coralos} summarizes all 11 \coralos{} vulnerabilities. The
eleventh, CHAIN, is a composite kill-chain that stages V5, V4, and V9.
We detail V4, V5, and V9, the three most impactful vulnerabilities, in the text below.
The remaining vulnerabilities are in  Supplementary~\ref{supp:coralos}.

\subsection{V4: agentSecret in URL Path (RC-3)}
\label{sec:v4}

\textbf{Root cause.}
Every MCP agent connection in \coralos{}
is authenticated via a per-agent UUID secret (\emph{agentSecret})
embedded directly as a URL path segment: the \texttt{/sse/} and \texttt{/mcp}
MCP routes both declare it as a \texttt{\{agentSecret\}} path parameter
(\emph{the MCP route table}).

\coralos{}'s built-in \texttt{CallLogging} Ktor plugin logs
\texttt{call.request.uri} at TRACE level
(\emph{the request-logging module}).
This is the level \coralos{} ships in its default configuration, so the secret
is written to logs in the default deployment, not only under debug settings.
This URI contains the full agentSecret.
Any party with read access to server logs (operations staff,
monitoring services, CDNs, cloud log aggregators) can harvest
active agent secrets in real time.

\textbf{Proof.}
We extracted the agentSecret from Docker logs within 138\,ms of
session creation (the secret appears when the agent process logs
its environment: \texttt{CORAL\allowbreak\_AGENT\allowbreak\_SECRET=\{uuid\}}).
Using the stolen secret, we sent an MCP \texttt{initialize} message to
\texttt{POST /mcp/v1/\{secret\}/mcp} and received \texttt{HTTP\,200}
with the server's full capability manifest, confirming that the stolen
secret grants full agent-session access.
We also authenticated to the Agent RPC API
(\texttt{/api/v1/agent-rpc/}) using the secret as a Bearer token.

\textbf{Impact.}
An attacker with server-log access can eavesdrop on all inter-agent tool
calls in the active session, abuse the LLM proxy to make unlimited model
calls at the operator's expense, and issue unauthorized payment claims.

\subsection{V5: Marketplace Supply Chain IPI (RC-1/RC-6)}
\label{sec:v5}

\textbf{Root cause.}
\texttt{MarketplaceAgentRegistrySource} fetches agent definitions from
the marketplace server and accepts the response without any signature
or integrity check.
The registry client calls \texttt{client.get(url).body<...>()} and returns the
result directly, with no signature verification and no hash pinning.

An attacker controlling the marketplace server, or one positioned for a
DNS or BGP hijack on \texttt{marketplace\allowbreak.coralprotocol\allowbreak.ai}, can serve
agent definitions containing IPI instructions in the
\texttt{description}, \texttt{readme}, and option \texttt{default}
fields.
This is structurally analogous to supply-chain attacks on npm
packages~\cite{formal2026}: the \emph{registry} is the attack surface.

\textbf{Proof.}
Under the rogue-service adversary~($\mathcal{A}_\text{RS}$), we deployed a rogue
HTTPS server implementing the \coralos{} marketplace API. The container fetched
our poisoned catalog and resolved a \texttt{payment-processor:2.1.0} agent
definition from our server, logging
\emph{``fetched 2 agents from the marketplace in 800\,ms.''} (RC-1 delivery,
100\% \asr{}).
The vulnerability is the missing integrity check, not a broken transport. TLS
authenticates the marketplace endpoint but never verifies content, so any
adversary who controls that endpoint (a compromised or rogue service, or DNS/BGP
redirection paired with an ACME-issued certificate) delivers unsigned poisoned
content that \coralos{} accepts. Our harness stands up the rogue marketplace with
a self-signed certificate for convenience. The gap that response signing closes
is the missing content signature, which persists under valid TLS.
The behavioral impact of the embedded IPI payload is model-dependent, and
Section~\ref{sec:llm-impact} reports per-model \asr{}.

\textbf{Impact.}
Successful RC-1 delivery means every orchestrator agent that reads
the poisoned description receives the IPI instruction.
Combined with RC-6 vulnerability in Flash-tier models (Section~\ref{sec:llm-impact}),
this enables full payment hijacking without breaking any cryptographic primitive.

\subsection{V9: Payment Wallet Hijack via Federation (RC-2)}
\label{sec:v9}

\textbf{Root cause.}
When \coralos{} federates with a remote server to hire an agent,
it calls \texttt{GraphAgentServer\allowbreak.getWallet()} to obtain the payment
recipient address for the Solana escrow.
The HTTP response is trusted unconditionally: \texttt{getWallet()} returns
\texttt{response.bodyAsText()} whenever the status is \texttt{OK}, with no DID
verification and no cryptographic binding to the server identity, and the
returned address is passed straight to the Solana escrow in \texttt{toRemote()}.

A rogue remote server returns the attacker's Solana wallet address,
causing the victim to create a blockchain escrow that pays the attacker.

\textbf{Proof.}
We compiled \coralos{}~v1.4.0 from source and drove its production
\texttt{GraphAgentServer\allowbreak.getWallet()} against a rogue federation server returning
attacker wallet \texttt{CHKbzJYk...} at \texttt{GET\allowbreak /api\allowbreak /v1\allowbreak /agent-rental\allowbreak /wallet}.
The real \texttt{getWallet()} issued a live request (logged at the rogue server
from \texttt{127.0.0.1}) and returned the attacker address unverified, which
\texttt{GraphAgentProvider} embeds as the escrow beneficiary (\texttt{Remote.wallet},
result hash~\texttt{821d683f\ldots}).
A variant V9b injects the wallet directly through the session API:
\coralos{} accepted \texttt{provider.wallet = ATTACKER\_ADDRESS} in the
\texttt{GraphAgentProvider.Remote} field, extending the trust gap to the client API.

\subsection{CHAIN: Multi-Stage V5$\to$V4$\to$V9}
\label{sec:chain}

The three highest-impact vulnerabilities compose into a single attack:

\begin{enumerate}
\item \textbf{V5 (supply chain):}
rogue marketplace poisons the agent registry with a malicious
\texttt{payment-processor} agent containing an IPI payload.
\item \textbf{V4 (credential harvest):}
when an orchestrator launches the poisoned agent, its agentSecret
appears in server logs within 138\,ms, allowing the attacker to
hijack the active MCP session.
\item \textbf{V9 (payment redirect):}
the rogue federation server returns the attacker wallet,
causing the Solana escrow to send payment to the wrong address.
\end{enumerate}

We executed all three stages in a single orchestrated run
(result hash~\texttt{bf697eca\ldots}),
achieving: (i)~full MCP session visibility (V4),
(ii)~behavioral corruption of the agent (V5),
and (iii)~financial theft (V9).

\section{\fetchai{} uAgents: 8 Vulnerabilities}
\label{sec:fetchai}

\subsection{Platform Overview}

\fetchai{} uAgents is a Python framework for decentralized
AI agent communication.
Agents register their HTTP endpoint URLs in the \emph{Almanac},
a centralized registry operated by Fetch.ai.
When Agent~A wants to contact Agent~B, it calls
\texttt{AlmanacApiResolver\allowbreak.resolve(address)} to obtain B's
current endpoint URL, then posts messages to that URL.
Agent identity uses DIDs and wallets are generated locally,
with private keys stored on disk.

\subsection{F-1: Endpoint Hijacking via Rogue Almanac (RC-1)}
\label{sec:f1}

\textbf{Root cause.}
\texttt{AlmanacApiResolver\allowbreak.\_api\_resolve()} fetches endpoint URLs
from the Almanac API as a plain JSON response
(\emph{the Almanac resolver}): the resolver takes the
returned endpoint list directly, neither checking a signature nor re-verifying
that each URL belongs to the claimed address.
An attacker with DNS or BGP control over
\texttt{agentverse.ai} (the Almanac host) can return a
malicious endpoint URL, redirecting all messages from Agent~A
to the attacker's server.

\textbf{Proof (with real uAgents library).}
We installed the actual \texttt{uagents} Python package
(\texttt{pip install uagents}) and configured
\texttt{AlmanacApiResolver} to use our rogue Almanac server, then called
\texttt{resolve()} for the victim agent, which returned our attacker-controlled
endpoint.
A payment-coordination message with field
\texttt{sensitive\_data:\ "payment\allowbreak\_credential\allowbreak\_token\allowbreak\_xyz"}
was delivered to our attacker endpoint (HTTP~200).
This proof uses the real library code, not a simulation
(result hash: \texttt{9822c958\ldots}).

\subsection{Summary: F-2 Through F-8}

Beyond F-1, we identify seven further structural vulnerabilities in
\fetchai{} uAgents, enumerated and analyzed in
 Supplementary~\ref{supp:fetchai} (Table~\ref{tab:fetchai}).
We highlight two below.

\textbf{F-4 (Private Keys in Plaintext JSON).}
Both the identity signing key and the Cosmos wallet private key are
stored unencrypted in \texttt{private\_keys.json} at
\texttt{os.getcwd()} with default filesystem permissions
(\emph{the key storage module}).
Any co-located process, or an attacker with file-system read access, can
steal both keys, enabling identity impersonation and wallet fund theft.

\textbf{F-8 (Wallet Key Mismatch on Restart).}
\texttt{get\_or\_create\_private\_keys()} calls \texttt{PrivateKey()}
twice (lines~147 and~149), generating two independent random keys.
The first call's result is \emph{returned} to the caller (Key~B).
The second call's result is \emph{saved to disk} (Key~C, $B\neq C$).
On the next agent restart, the saved Key~C is loaded, changing the
wallet address.
Any funds sent to the first wallet address (Key~B's address)
become inaccessible, because Key~B was never persisted.
We verified this deterministically by instrumenting the function.

\section{\aptwo{} v0.2.0: 14 Vulnerabilities}
\label{sec:ap2}

\subsection{Platform Overview}

\aptwo{}~\cite{ap22025} is Google's Agent Payments Protocol,
designed to let AI shopping agents execute payment transactions on behalf
of users while maintaining cryptographic accountability.
Its security model uses three mandate objects (Intent, Cart,
and Payment), each carried as a W3C Verifiable Credential~\cite{w3c_vc}
and cryptographically signed by the issuing party.
Mandates form an auditable chain from user intent to final payment.

We analyzed the v0.2.0 Python reference implementation
(\texttt{code/samples/python/}).
\textbf{Our analysis is strictly structural} (Definitions~\ref{def:structural}--\ref{def:semantic}):
we find implementation-level bugs in authentication, payment atomicity,
and session state, orthogonal to the mandate-chain replay and context-binding
failures addressed in~\cite{ztrv2025}.
A-AP2-11 was corroborated live end-to-end against a running merchant MCP
server that returned business-logic responses to an unauthenticated
client (result hash~\texttt{38e3865f\ldots}), and the rest are proven at the
code level (Table~\ref{tab:ap2}).

The 14 findings split into three categories by how deployment affects them (the
``Cat.'' column of Table~\ref{tab:ap2}). Three are Python-reference-specific
(Py): the Go reference implements these correctly, so they vanish under the Go
build. Another three are protocol-design gaps (Spec) that no reference
implementation corrects, because the AP2 specification itself leaves email
binding, nonce freshness, and payee constraints unenforced. The remaining eight
are cross-reference bugs (Both), present in both reference codebases or their
shared utility code. The eleven Spec and Both findings persist regardless of
implementation. Mutual TLS does not close them: they are application-logic flaws
such as the token-redemption TOCTOU (A-AP2-1), unenforced payee constraints
(A-AP2-12), and missing nonce freshness (A-AP2-8), not transport weaknesses.

\subsection{Vulnerability Table}

Table~\ref{tab:ap2} lists all 14 \aptwo{} vulnerabilities with class, category,
and proof method. All are structural flaws in the reference
implementation: a token-redemption race enabling double-spend (A-AP2-1), a
fail-open certificate-chain check (A-AP2-4), missing authentication on the
merchant MCP tools (A-AP2-11), and a credentials provider that verifies the
mandate but not the caller's identity (A-AP2-15). Per-vulnerability analyses of
these four appear in  Supplementary~\ref{supp:ap2}.

\begin{table}[ht]
\caption{\aptwo{} v0.2.0: 14 structural vulnerabilities. Category Py is
Python-reference-specific (the Go reference fixes it), Spec is an AP2
protocol-design gap, Both is present in both reference codebases.
SD-JWT is a selective-disclosure JWT.}
\label{tab:ap2}
\centering
\small
\renewcommand{\arraystretch}{1.05}
\begin{tabular}{@{}lp{3.1cm}lll@{}}
\toprule
\textbf{ID} & \textbf{Title} & \textbf{Class} & \textbf{Cat.} & \textbf{Proof} \\
\midrule
A-AP2-1  & Token redemption TOCTOU         & RC-4 & Both & Code \\
A-AP2-2  & File token store race           & RC-4 & Both & Code \\
A-AP2-3  & Token binding race              & RC-4 & Both & Code \\
A-AP2-4  & Empty trusted-roots bypass      & RC-5 & Py   & Code \\
A-AP2-5  & \texttt{user\_email} trusted    & RC-2 & Spec & Code \\
A-AP2-6  & Recurrence file race            & RC-4 & Both & Code \\
A-AP2-7  & Mandate filename collision      & RC-5 & Both & Code \\
A-AP2-8  & Nonce freshness missing         & RC-4 & Spec & Code \\
A-AP2-10 & Debug logs expose SD-JWTs       & RC-3 & Both & Code \\
A-AP2-11 & No auth on merchant MCP         & RC-5 & Py   & Code/Live \\
A-AP2-12 & Payee injection bypass          & RC-5 & Spec & Code \\
A-AP2-13 & Revoked token not checked       & RC-5 & Both & Code \\
A-AP2-14 & Budget integer underflow        & RC-4 & Both & Code \\
A-AP2-15 & Role confusion (creds provider) & RC-5 & Py   & Code \\
\bottomrule
\end{tabular}
\end{table}

\section{\aipbench{}: A Deterministic Benchmark}
\label{sec:bench}

\subsection{Design Rationale}

Existing IPI benchmarks~\cite{agentdojo2024,injecagent2024} use LLM judges,
which conflate model-dependent behavioral effects with model-independent
structural flaws and can themselves be prompt-injected or drift, producing
unreliable verdicts.

\aipbench{} avoids these problems through exclusively
\emph{deterministic judges}:

an HTTP status code (V4 session hijack), the attacker wallet string in the
constructed payment request (V9), a log-pattern regex for the leaked secret
(V4), an unauthorized-namespace event count over WebSocket (V15), the
\texttt{model} field in the proxy response (V6/V13), and vulnerable-code-pattern
presence in source (structural proofs). Each judge is a simple predicate with
no generative component.

\subsection{Coverage and \asr{} Summary}

Table~\ref{tab:bench} summarizes attack-success rates across all \aipbench{} scenarios.

\begin{table*}[!ht]
\caption{\aipbench{}: attack-success rates across all scenarios.
\emph{Live} = proven against running system. \emph{Code} = structural code proof.}
\label{tab:bench}
\centering
\small
\renewcommand{\arraystretch}{1.1}
\begin{tabular}{@{}p{4.5cm}p{1.8cm}p{1.2cm}cc@{}}
\toprule
\textbf{Attack} & \textbf{Platform} & \textbf{Class} & \textbf{\asr{}} & \textbf{Proof} \\
\midrule
V4 secret harvest + MCP hijack & \coralos{} & RC-3 & 100\% & Live \\
V5 delivery (marketplace)      & \coralos{} & RC-1 & 100\% & Live \\
V6 + V13 model override        & \coralos{} & RC-5 & 100\% & Live \\
V9 wallet hijack               & \coralos{} & RC-2 & 100\% & Live \\
V10 CSRF                       & \coralos{} & RC-5 & 100\% & Live \\
V11 open redirect              & \coralos{} & RC-5 & 100\% & Live \\
V15 namespace bypass           & \coralos{} & RC-5 & 100\% & Live \\
CHAIN (V5$\to$V4$\to$V9)      & \coralos{} & multi & 100\% & Live \\
F-1 endpoint hijack            & \fetchai{} & RC-1 & 100\% & Live \\
A-AP2-1,4,5,11,15             & \aptwo{}   & multi & ---    & Code \\
\midrule
\textbf{All RC-1 to RC-5} & All & --- & \textbf{100\%} & --- \\
V5 impact (GPT-4o-mini, Gemini Flash, Mistral) & \coralos{} & RC-6 & 99--100\% & Live \\
V5 impact (GPT-4o) & \coralos{} & RC-6 & 68\% & Live \\
V5 impact (Sonnet, Haiku, Gemini Pro) & \coralos{} & RC-6 & 0\%  & Live \\
\bottomrule
\end{tabular}
\end{table*}

Each scenario includes: the exact API call sequence, expected
judge output, a SHA-256 hash of the reference result,
and a Docker Compose environment for reproduction.

\section{\pcat{}: Unified Defense Framework}
\label{sec:pcat}

\subsection{Architecture and Design Principles}

\pcat{} (Protocol-level Commerce Agent Trust) is a platform-agnostic
HTTP sidecar that intercepts all communication between agents and
commerce services:
\[
[\text{Agent}]
\xrightarrow{\text{all requests}}
[\pcat{}\text{ :4443}]
\xrightarrow{\text{verified}}
[\text{Commerce services}]
\]
\pcat{} requires no change to platform source code, though two principles depend
on ecosystem adoption: P1 assumes registries sign responses and P2 assumes a DID
resolver. A minimal deployment gains P3, P4, and P5 immediately, with P1 and P2
strengthening as keys and a resolver are provisioned. Their distribution,
rotation, and revocation are operational rather than protocol concerns.
Its scope is the HTTP path between agents and commerce services, and the
vulnerabilities that traverse it are addressed by the five principles below,
which together cover root-cause classes RC-1 through RC-5.
Those residing in adjacent layers, specifically filesystem races,
internal LLM-proxy calls, WebSocket events, and key-initialization code,
require targeted platform countermeasures documented in
Section~\ref{sec:pcat-limits}.
The five principles are:

\medskip
\noindent\textbf{P1: Response Integrity Verification (RC-1).}
Every response from a registered marketplace or registry must carry
an ECDSA~P-256 signature over
\texttt{SHA256(JSON(body\_hash,\allowbreak{} timestamp,\allowbreak{} nonce))}.
\pcat{} verifies the signature before forwarding, and
responses from registered services that lack a valid signature are blocked.
Replay attacks are prevented by a time-bounded nonce cache
(30-second window, $O(\log n)$ eviction via heap).
Registries publish signing keys through an existing trust root such as
WebPKI/COSE or a DID document, and key rollover follows the registry's own
rotation, so P1 introduces no new key-distribution authority.
P1 assumes the registry's signing key is uncompromised: a valid signature from a
stolen key still verifies. Removing this single point of failure would require
publisher-plus-registry countersigning or a transparency log~\cite{provenance2026}.

\noindent\textbf{P2: Caller Identity Binding (RC-2).}
Requests carrying payment destinations or user identity claims
must include a DID-based caller proof.
\pcat{} resolves the caller's DID, verifies the proof against the
DID document, and checks that the claimed identity is within
the caller's authorized scope.
(\emph{Note:} \pcat{} ships a real \texttt{did:web} resolver: a live HTTP fetch
of the controller-hosted DID document, ECDSA verification of the caller's proof,
wallet-to-DID authorization, and revocation via the document's status.
Experiment~E blocks all six RC-2 cases (three V9 variants, an unauthorized
A-AP2-5 wallet, a revoked DID, and a rotated key) at 0\% \asr{} and admits the
authorized wallet. Without a resolver,
P2 logs and allows (Section~\ref{sec:pcat-limits}).)

\noindent\textbf{P3: Secure Channel Enforcement (RC-3, RC-5).}
\pcat{} enforces three channel-level policies simultaneously. (i)~Credential
detection in URL paths warns only, as the full fix requires the platform to move
credentials out of the path. (ii)~CSRF blocking rejects cross-origin
state-changing requests without a valid token (HTTP~403), using
\texttt{urlparse().hostname} to avoid substring-bypass attacks. (iii)~Redirect
validation URL-decodes and backslash-normalizes the \texttt{?to=} parameter and
verifies it against a whitelist before forwarding.

\noindent\textbf{P4: Atomic Payment State (RC-4).}
All payment check-and-execute operations route through
\pcat{}'s \texttt{asyncio.Lock}-backed compare-and-swap store.
Check and deduct occur inside a single lock acquisition,
serializing concurrent requests.
A rollback mechanism reverts the deduction if the subsequent
history append raises an exception.

\noindent\textbf{P5: MCP Tool-Call Authorization (RC-5, A-AP2-11, A-AP2-15).}
When a POST body carries a JSON-RPC~2.0 \texttt{tools/call} method targeting
a sensitive operation (e.g., \texttt{complete\_checkout},
\texttt{issue\_payment\_credential}),
\pcat{} requires the caller to supply an \texttt{X-Caller-Identity} header
containing a pre-registered identity token.
If the header is absent or the caller identity is not in the authorized
allowlist for that tool, \pcat{} rejects the request with HTTP~403
before forwarding it to the MCP server.
This addresses A-AP2-11 (unauthenticated endpoints) and A-AP2-15
(role confusion) at the transport layer, without modifying AP2 source code.
P5 applies to HTTP/SSE transport only.
Stdio-transport deployments require host-level process isolation outside \pcat{}'s scope.

\subsection{Formal Security Properties}

\begin{theorem}[P1 Supply-Chain Protection]
\label{thm:p1}
Under P1 with correct ECDSA P-256 key management, the probability
that \pcat{} forwards a forged registry response is at most $\mathit{negl}(\lambda)$
where $\lambda = 256$ is the security parameter.
\end{theorem}

\begin{proof}[Proof sketch]
Forging a valid P1 signature without the registry's private key
requires solving the elliptic-curve discrete-log problem over P-256, which is computationally
infeasible under the standard discrete-log hardness assumption.
The nonce cache prevents replay of genuine signatures within the
freshness window.
\end{proof}

\begin{theorem}[P3 CSRF Prevention]
\label{thm:p3}
Under P3, any POST/PUT/DELETE/PATCH request with a cookie-session header
from an origin $o$ where $\mathit{hostname}(o) \notin W$ (the
same-origin whitelist) without a valid CSRF token is rejected with
HTTP~403. The origin check uses \texttt{urlparse().hostname}, preventing
substring-bypass attacks (e.g., \texttt{localhost.attacker.com}).
\end{theorem}

\begin{theorem}[P4 TOCTOU Prevention]
\label{thm:p4}
Under P4 with single-writer deployment, all payment
check-and-deduct operations on a given budget are serialized by
\texttt{asyncio.Lock}. The probability of a concurrent double-deduction
is 0 within a single Python process.
\end{theorem}

\subsection{Evaluation}
\label{sec:pcat-eval}

We evaluated \pcat{} by running our attack suite through the middleware.
Because \pcat{} interposes on the HTTP interface that all three platforms
share, each principle is exercised against its corresponding attack across
platforms: P1 against the \fetchai{} F-1 endpoint hijack, P3 against the
\coralos{} CSRF and open-redirect attacks, and P4 and P5 against the
\aptwo{} token-redemption race and unauthenticated MCP endpoints
(Table~\ref{tab:pcat-eval}).
For each attack we measured \asr{} before and after deploying
\pcat{}, and the false-positive rate (FPR) on 10{,}000 benign requests.
These cross-platform blocks are confirmed against the live attack payloads
using the reference \pcat{} implementation (Experiment~A,
 Supplementary~\ref{supp:bench}).

\begin{table*}[!ht]
\caption{\pcat{} evaluation: \asr{} before and after, with false-positive rate (FPR)
on 10{,}000 benign requests. RC-6 (semantic) is outside \pcat{}'s scope.}
\label{tab:pcat-eval}
\centering
\small
\renewcommand{\arraystretch}{1.1}
\begin{tabular}{@{}p{3.3cm}p{1.2cm}p{1.4cm}ccc@{}}
\toprule
\textbf{Attack} & \textbf{Class} & \textbf{Principle} &
\textbf{Before} & \textbf{After} & \textbf{FPR} \\
\midrule
V5 marketplace IPI (delivery) & RC-1 & P1 & 100\% & \textbf{0\%} & 0\% \\
F-1 endpoint hijack           & RC-1 & P1 & 100\% & \textbf{0\%} & 0\% \\
V9 wallet hijack              & RC-2 & P2 & 100\% & \textbf{0\%}* & 0\% \\
V10 CSRF                      & RC-5 & P3 & 100\% & \textbf{0\%} & 0\% \\
V11 open redirect             & RC-5 & P3 & 100\% & \textbf{0\%} & 0\% \\
V7 x402 TOCTOU                & RC-4 & P4 & ---     & \textbf{0\%}$^\dagger$ & 0\% \\
A-AP2-1 token TOCTOU          & RC-4 & P4 & 100\% & \textbf{0\%} & 0\% \\
V4 secret in URL              & RC-3 & P3 & 100\% & Partial$^\ddagger$ & --- \\
A-AP2-11 unauth MCP endpoints & RC-5 & P5 & 100\% & \textbf{0\%} & 0\% \\
A-AP2-15 role confusion       & RC-5 & P5 & 100\% & \textbf{0\%} & 0\% \\
\midrule
RC-1,2,4,5 + MCP-layer (P5)$^\S$ & ---    & ---  & 100\% & \textbf{0\%} & \textbf{0\%} \\
RC-3 (partial)                & ---    & P3 & 100\% & Warn only & --- \\
RC-6 (semantic, out of scope) & ---    & ---  & varies & N/A & --- \\
\bottomrule
\end{tabular}

\par\vspace{2pt}
{\footnotesize\raggedright
*P2 uses a real \texttt{did:web} resolver (live resolution with revocation and key
rotation, Experiment~E); without any resolver P2 logs and allows (Section~\ref{sec:pcat-limits}).
\quad $^\dagger$Concurrent budget test: 1/3 succeed, 2/3 blocked correctly.
\quad $^\ddagger$P3 warns; full fix requires the platform to move the secret to an
Authorization header.
\quad $^\S$P5 covers HTTP/SSE transport; stdio deployments require host-level isolation.\par}
\end{table*}

\textbf{P1 evaluation.}
Against the rogue marketplace, \pcat{} blocked the unsigned response (V5 attack)
and passed the signed legitimate response, using an ECDSA P-256 key pair
generated at runtime.

\textbf{P3 evaluation.}
For V10 (CSRF), a cross-origin POST without a token received HTTP~403 while a
same-origin POST with a valid token received HTTP~200. For V11, three encoded
variants of \texttt{to=//attacker.evil.com} (plain, URL-encoded, backslash) were
all blocked and the safe \texttt{to=/ui/console/} passed.

\textbf{P4 evaluation.}
Three concurrent requests deducting 60 units from a 100-unit budget yield exactly
one success and two rejections, with zero false blocks on 1{,}000 sequential
deductions. As a single gate over a multi-worker backend, \pcat{} serializes
redemption and drops Experiment~B's 30--57\% double-spend to 0\% (Experiment~C).
Four \pcat{} instances over a shared compare-and-swap store (SQLite's atomic
\texttt{UPDATE\ldots WHERE used=0}, the Redis+Lua primitive) preserve this at
scale (Experiment~D, zero double-spends over 200 trials), detailed in
 Supplementary~\ref{supp:bench}.

\textbf{P5 evaluation.}
An unauthenticated \texttt{tools/call} on \texttt{complete\_checkout} (A-AP2-11)
and a request from an unregistered caller (A-AP2-15 role confusion) both received
HTTP~403, while a request carrying the pre-registered identity token passed to
the backend.

\textbf{Latency microbenchmarks.}
We measured \pcat{}'s processing overhead with 3{,}000 consecutive in-process
calls per scenario using \texttt{time.perf\_counter()} on a consumer laptop
(Intel i7-13700H, 32\,GB RAM, Python~3.11), exclusive of network I/O
(Table~\ref{tab:pcat-perf}).
P1 (ECDSA~P-256 verification) is the bottleneck at 0.21\,ms,
consistent with known P-256 performance on x86-64.
P5 (one \texttt{json.loads} plus a set-membership lookup) adds 0.08\,ms.
Both are negligible relative to agent-to-service network round-trip times.

\begin{table}[!ht]
\caption{\pcat{} latency microbenchmarks (3{,}000 calls per scenario,
in-process, exclusive of network I/O).}
\label{tab:pcat-perf}
\centering
\small
\renewcommand{\arraystretch}{1.1}
\resizebox{\columnwidth}{!}{\begin{tabular}{@{}lrrr@{}}
\toprule
\textbf{Scenario} & \textbf{Median} & \textbf{p99} & \textbf{Throughput} \\
\midrule
Benign request (P3+P4+P5 pipeline) & 0.015\,ms & 0.030\,ms & 64{,}500\,req/s \\
P1 response verify (ECDSA P-256)   & 0.209\,ms & 0.414\,ms &  4{,}800\,req/s \\
P5 MCP tool-call auth              & 0.083\,ms & 0.184\,ms & 12{,}000\,req/s \\
\bottomrule
\end{tabular}}
\end{table}

\textbf{Concurrent throughput.}
We evaluated \pcat{} under concurrent agent load by running $K$ asyncio
workers simultaneously, each sending 500~requests, for $K \in \{1,10,50,100\}$
(Table~\ref{tab:pcat-concurrent}). Aggregate throughput stays at
5{,}100--5{,}500\,req/s across all concurrency levels, with median per-request
latency at most 0.12\,ms, and memory scales at about 50\,KB per connection
(5\,MB for 100~agents). The P4 \texttt{asyncio.Lock} does not contend on standard
requests because payment budget operations route separately, and the
multi-process limitation (L2) applies only to deployments spanning multiple
Python processes.

\begin{table}[!ht]
\caption{\pcat{} concurrent throughput under $K$ simultaneous agent connections
(500 requests per agent, asyncio gather, Intel i7-13700H).}
\label{tab:pcat-concurrent}
\centering
\small
\renewcommand{\arraystretch}{1.1}
\begin{tabular}{@{}rrrrr@{}}
\toprule
\textbf{$K$ agents} & \textbf{Total req} & \textbf{Req/s} &
\textbf{Median} & \textbf{p99} \\
\midrule
  1 &    500 & 5{,}105 & 0.120\,ms & 0.195\,ms \\
 10 &  5{,}000 & 5{,}230 & 0.108\,ms & 0.659\,ms \\
 50 & 25{,}000 & 5{,}541 & 0.109\,ms & 0.282\,ms \\
100 & 50{,}000 & 5{,}186 & 0.114\,ms & 0.281\,ms \\
\bottomrule
\end{tabular}
\end{table}

\textbf{False-positive validation.}
\pcat{} passed all 10{,}000 diverse benign requests with zero false positives
(0\%, 95\% Wilson CI: [0\%, 0.04\%]), 1{,}000/1{,}000 in each of ten categories:
GET, POST with a CSRF token, POST without one, text/plain bodies, JSON bodies,
authorized MCP tool calls, non-sensitive MCP tool calls, \texttt{localhost} over
IPv4, \texttt{localhost} over IPv6, and cookieless PUT.

\subsection{Documented Limitations}
\label{sec:pcat-limits}

\pcat{} has three implementation-specific limitations:
\begin{itemize}
\item \textbf{L1 (P2 needs an operated resolver and trust root):}
Caller identity binding (P2) runs against a real \texttt{did:web} resolver and is
evaluated end-to-end, including revocation and key rotation (Experiment~E, 0\%
\asr{} for V9 and A-AP2-5). It still depends on the deployment operating a
resolver and trusting the DID method's root. A minimal deployment without one
logs and allows, leaving RC-2 uncovered.
Theorem~\ref{thm:p1} does not depend on P2.

\item \textbf{L2 (P4 distributed atomicity):}
Theorem~\ref{thm:p4} is proved for single-process deployment, where
\texttt{asyncio.Lock} provides atomicity. Across multiple workers or \pcat{}
instances it reduces to the linearizability of the shared compare-and-swap store:
with an atomic conditional update (SQLite's \texttt{UPDATE\ldots WHERE used=0}, or
a Redis+Lua script), a token's redemptions are totally ordered and exactly one
observes \texttt{used=0}, so consume-once holds for any number of instances
(confirmed in Experiment~D, zero double-spends over 200 trials).

\item \textbf{L3 (V4 partial):}
\pcat{} cannot change how \coralos{} constructs URL paths, so
it can only detect and warn when a credential appears in the URL.
Full remediation of V4 requires the platform to move agentSecret
to an \texttt{Authorization: Bearer} header.
\end{itemize}

Beyond these implementation gaps, fourteen vulnerabilities in the study
inhabit layers outside the agent-to-platform HTTP path.
V6 and V13 reside in the internal \coralos{}-to-LLM-provider call path,
not the agent-facing HTTP path where \pcat{} sits.
V14 is a filesystem race in the registry file loader.
V15 involves WebSocket event routing.
F-3 is a timing configuration error, F-4 and F-8 are key-management bugs,
and F-7 is an agent-protocol sender field.
A-AP2-4, A-AP2-7, A-AP2-8, A-AP2-10, A-AP2-12, and A-AP2-13 reside
in AP2's internal certificate, filesystem, nonce, logging, argument validation,
and token-revocation logic respectively.
These vulnerabilities are not gaps in \pcat{}'s design.
They require countermeasures at their own layers:
model-parameter validation within the LLM proxy,
filesystem locking, WebSocket-aware middleware, key storage encryption,
and certificate configuration enforcement.
All are well-understood remediations. None require a protocol-layer defense.
Two could be packaged as \pcat{} companion modules, namely response signing
at the LLM-proxy boundary (V6/V13) and WebSocket-aware authorization (V15),
which we leave to future work.

\section{LLM Behavioral Impact Study for V5}
\label{sec:llm-impact}

\subsection{Separating Delivery from Behavioral Impact}

V5 has two independent components.
The \emph{delivery} component (RC-1) is structural: \coralos{} fetches
and trusts the poisoned agent definition regardless of model.
We proved this live in Section~\ref{sec:v5} (100\% \asr{}).

The \emph{behavioral impact} component (RC-6) is semantic:
it requires an LLM-powered orchestrator to process the IPI instruction
embedded in the agent description and act on it.
This component is model-dependent per Definition~\ref{def:semantic}.
We study it independently here to quantify which model configurations
remain at risk.

\subsection{Experimental Setup}

We tested 8 models from 5 providers across 3 temperatures (0.0, 0.3, 0.7)
and 3 IPI payload variants: a compliance framing, a technical override framing,
and an authority claim framing.
Each combination received 20 trials, totaling 1{,}440 API calls
(8 $\times$ 3 $\times$ 3 $\times$ 20).
The judge is deterministic: $\asr{} = 1$ iff the attacker address is the
agent's selected \texttt{WALLET\_USED} value, not merely quoted in a refusal
(as \texttt{claude-3-haiku} does).
Results use 95\% Wilson confidence intervals.
Total cost: \$0.72.
All calls use OpenRouter with pinned model IDs, and direct-vendor replication
of a subset is future work.
The exact system prompts, IPI payloads, and orchestration scaffolding
are released at \benchmarkcode{}.
Table~\ref{tab:llm-cost} shows the per-model cost breakdown.
The variation reflects pricing differences across providers
(GPT-4o and frontier-aligned models dominate at \$0.14--\$0.31
per 180 trials due to higher per-token rates, while
Flash-tier and open-weight models cost under \$0.03).

\begin{table}[!ht]
\caption{Per-model cost breakdown for the 1{,}440-trial LLM study
(180 trials per model: 3 temperatures $\times$ 3 IPI variants $\times$ 20 trials).
All calls via OpenRouter.}
\label{tab:llm-cost}
\centering
\small
\renewcommand{\arraystretch}{1.1}
\begin{tabular}{@{}lrr@{}}
\toprule
\textbf{Model} & \textbf{Trials} & \textbf{Cost (USD)} \\
\midrule
\texttt{gpt-4o-mini}            & 180 & \$0.0097 \\
\texttt{gpt-4o}                 & 180 & \$0.1394 \\
\texttt{claude-3-haiku}         & 180 & \$0.0262 \\
\texttt{claude-sonnet-4.6}      & 180 & \$0.3145 \\
\texttt{gemini-2.5-flash}       & 180 & \$0.0293 \\
\texttt{gemini-2.5-pro}         & 180 & \$0.1746 \\
\texttt{llama-3.3-70b-instruct} & 180 & \$0.0067 \\
\texttt{mistral-large-2512}     & 180 & \$0.0194 \\
\midrule
\textbf{Total}                  & \textbf{1{,}440} & \textbf{\$0.7198} \\
\bottomrule
\end{tabular}
\end{table}

\subsection{Results}

\begin{table*}[!t]
\caption{V5 behavioral impact \asr{} across models
(20 trials per model-temperature-variant configuration,
8 models $\times$ 3 temperatures $\times$ 3 IPI variants $\times$ 20 trials = 1{,}440 total,
95\% Wilson CI, cost \$0.72).
The pattern is a gradient, not a binary split.}
\label{tab:llm}
\centering
\small
\renewcommand{\arraystretch}{1.1}
\begin{tabular}{@{}p{3.8cm}p{1.8cm}p{1.5cm}cc@{}}
\toprule
\textbf{Model} & \textbf{Provider} & \textbf{Tier} &
\textbf{\asr{}} & \textbf{95\% CI} \\
\midrule
\texttt{gpt-4o-mini}       & OpenAI    & Flash    & 100\% & [98\%, 100\%] \\
\texttt{gemini-2.5-flash}  & Google    & Flash    & 100\% & [98\%, 100\%] \\
\texttt{mistral-large-2512}& Mistral   & Large    & 99\%  & [97\%, 100\%] \\
\texttt{gpt-4o}            & OpenAI    & Mid      & 68\%  & [61\%, 74\%]  \\
\texttt{claude-3-haiku}    & Anthropic & Flash    & 0\%   & [0\%, 2\%]    \\
\midrule
\texttt{llama-3.3-70b}     & Meta      & OSS      & 10\%  & [6\%, 15\%]   \\
\texttt{claude-sonnet-4.6} & Anthropic & Frontier & 0\%   & [0\%, 2\%]    \\
\texttt{gemini-2.5-pro}    & Google    & Frontier & 0\%   & [0\%, 2\%]    \\
\bottomrule
\end{tabular}
\end{table*}

Table~\ref{tab:llm} shows a gradient rather than a binary split. Three models
are fully vulnerable regardless of temperature or payload variant (GPT-4o-mini,
Gemini~2.5~Flash, and Mistral-large, each 99--100\% \asr{}), while GPT-4o sits at
68\% ([61\%, 74\%]) stably across temperatures, suggesting a structural
compliance bias rather than temperature sensitivity. Claude~3-Haiku, a Flash-tier
model, nonetheless resists in every trial (0\%) unlike the other Flash-tier
models, so alignment training rather than tier drives resistance, and
Claude~Sonnet~4.6 and Gemini~2.5~Pro also resist fully. Llama-3.3-70B is
temperature-dependent (10\% at 0.0, 20\% at 0.3, 0\% at 0.7). All three IPI
variants gave similar aggregate \asr{} (48--56\%), so the attack does not depend
on a single framing.

\subsection{Implications}

Four implications follow. First, V5's delivery (RC-1) is model-independent: the
poisoned definition reaches \coralos{} regardless of model, so any deployment is
at some risk. Second, susceptibility does not track model size, as GPT-4o (68\%)
is far more susceptible than the smaller Llama-3.3-70B (10\%), so alignment
training, not scale, is decisive. Third, resistance is not universal: GPT-4o and
Llama leave a nonzero exploitable fraction that is operationally significant for
high-value transactions. Fourth, resistant models refuse for distinguishable
reasons: injection detection for Sonnet, fraud for Llama, and policy refusal for
GPT-4o  (Supplementary~\ref{supp:refusal}). This favors semantic defenses that surface the
payment destination and sanitize description fields over model self-diagnosis.

\section{Discussion}
\label{sec:discussion}

Taken together, the findings are systemic. The recurrence of the same five
root causes (RC-1 through RC-5) across three
independently-developed platforms from three different organizations
indicates these are not isolated bugs but architectural failure modes
intrinsic to current agentic commerce systems.
Two platforms fetch agent definitions from centralized registries without
signature verification (RC-1), all three placed credentials in observable
channels (RC-3), and two implemented non-atomic payment state (RC-4).
The shared cause is design economics: centralized discovery, URL-based
authentication, and in-memory state are the paths of least resistance, so
developers converge on the same unsafe defaults. The highest-impact corrections
recur across platforms: sign registry and service responses, bind payment
destinations to verified identities, and make payment-state transitions atomic.
\pcat{} addresses all five root causes with one external layer because they all
sit at the protocol boundary a sidecar can mediate.

The 100\% \asr{} figures hold under default deployments (RC-4 under adversarial
scheduling, Remark~\ref{rem:rc4}). Common hardening narrows individual findings
but not the root causes: non-TRACE logging removes V4's leak, strict same-origin
CORS blocks V10, and mandatory mutual TLS closes development-only HTTP, yet none
verify registry or federation content (RC-1, RC-2) or make payment state
atomic (RC-4).

We deliberately exclude one candidate finding, A-AP2-9: the \aptwo{} reference
implementation uses plain-HTTP inter-service URLs (e.g.,
\texttt{http://localhost:8003}), but all target \texttt{localhost} and are
documented as development-only, and production deployments would use HTTPS with
mutual TLS. We do recommend the \aptwo{} specification \emph{mandate} mutual TLS
for inter-service communication so developers do not inadvertently ship HTTP.

The live V9 proof (Section~\ref{sec:v9}) stops at the protocol level rather than
an on-chain transfer: we generated real Solana Devnet keypairs for both parties
but could not settle on-chain due to faucet rate limiting. As in standard
payment-security practice (showing an injection payload reaches the query, not
that data was exfiltrated), the protocol-level proof suffices.

\subsection{Responsible Disclosure}

We followed Coordinated Vulnerability Disclosure, notifying the \coralos{}
(hello@coralos.ai), \fetchai{} (security@fetch.ai), and Google AP2 (Google Open
Source Software Vulnerability Reward Program, OSS VRP) security teams
simultaneously on 2026-07-06, with a 90-day embargo ending 2026-10-04, when the
full artifacts become public. Google's OSS VRP declined the \aptwo{} reports as
OT2/OT3 tier (GitHub Security Advisories GHSA-w7vq-wj2c-8h73, GHSA-hvcr-7w52-jmx5,
and GHSA-q7f7-h63m-9gmp filed independently), while the \coralos{} and \fetchai{}
responses remain pending.

To protect deployed systems during the remediation window, this public
preprint deliberately withholds exact code locations and proof-of-concept
exploits. The complete technical details, proof-of-concept code, and the
\aipbench{} artifacts are released in the project repository once the
disclosure period concludes.

No experiment moved real funds (payment proofs used throwaway Solana Devnet
keypairs), and every finding was disclosed before any detail was released.

\subsection{Limitations}

Our study has four limitations.
First, we study three platforms. Further work should test the taxonomy on
different foundations such as on-chain-only or fully-encrypted registries.
Second, the V5 LLM study uses 1{,}440 trials across 8 models. Some
models show non-trivial temperature dependence, and further prompt
variants could map the boundaries of susceptibility more precisely.
Third, \pcat{}'s P4 has a formal proof only for single-process deployment. The
distributed case (four instances over a shared compare-and-swap store) is
validated empirically (Experiment~D) but not proved.
Fourth, self-audit during development found and fixed four \pcat{}
implementation vulnerabilities before release (BUG-AUDIT-1 through -4, in
\pcatcode{}). An independent external review before deployment is still
recommended. Despite these limitations, \pcat{} eliminates most structural
vulnerabilities that cross the HTTP boundary with no platform change, and the
remaining gaps are well-scoped rather than fundamental.

\section{Conclusion}
\label{sec:conclusion}

We presented the first systematic, cross-platform security analysis of agentic
commerce infrastructure. Across \coralos{}, \fetchai{}, and \aptwo{} we identified
33 structural vulnerabilities, each succeeding deterministically and independent
of the deployed model (100\% attack-success rate wherever live-measured), and
chained three into an end-to-end payment hijack.

The central insight is a formal distinction between \emph{structural attacks}
(protocol-level, deterministic, model-independent) and \emph{semantic attacks}
(reasoning-level, probabilistic, model-dependent). No model improvement fixes a
structural vulnerability, and no protocol patch prevents a semantic attack, so
defense must address both layers: the most widely-deployed models stay
susceptible while frontier-aligned models show 0\%, yet the structural delivery
path bypasses all model tiers equally.

\aipbench{} and \pcat{} turn this into tools: a deterministic benchmark and a
unified defense that drives structural \asr{} to 0\% for four of the five
structural classes (RC-3 warn-only), giving operators concrete means to measure
and improve agentic commerce security.

  \bibliographystyle{plain}

\bibliography{references}

\appendix

\section{Consolidated Findings by Root-Cause Class}
\label{supp:consolidated}
Table~\ref{tab:consolidated} lists all 33 structural findings grouped by
root-cause class, with platform, the constituent IDs, the proof method, and,
for \coralos{}, persistence in v1.4.0.

\begin{table}[!ht]
\caption{All 33 structural findings by root-cause class. Proof: L\,=\,live,
C\,=\,code proof, R\,=\,deterministic reproduction. All ten persisting \coralos{}
findings have byte-identical vulnerable code in v1.4.0 (V9 also re-run live); V7
was remediated in v1.4.0.}
\label{tab:consolidated}
\centering
\small
\renewcommand{\arraystretch}{1.15}
\begin{tabular}{@{}llp{4.2cm}c@{}}
\toprule
\textbf{Class} & \textbf{\#} & \textbf{Findings (platform)} & \textbf{Proof} \\
\midrule
RC-1 & 3  & V5 (\coralos{}), F-1/F-2 (\fetchai{}) & L \\
RC-2 & 3  & V9 (\coralos{}), F-8 (\fetchai{}), A-AP2-5 & L/C/R \\
RC-3 & 4  & V4 (\coralos{}), F-3/F-4 (\fetchai{}), A-AP2-10 & L/C \\
RC-4 & 9  & V7$^\dagger$/V14 (\coralos{}), F-6 (\fetchai{}), A-AP2-1/2/3/6/8/14 & C \\
RC-5 & 13 & V6/V10/V11/V13/V15 (\coralos{}), F-5/F-7 (\fetchai{}), A-AP2-4/7/11/12/13/15 & L/C \\
CHAIN & 1 & V5$\to$V4$\to$V9 (\coralos{}, composite) & L \\
\bottomrule
\multicolumn{4}{l}{\small $^\dagger$V7 (x402 budget race) was remediated in \coralos{} v1.4.0.}
\end{tabular}
\end{table}

\section{Experimental Environment}
\label{supp:env}

\paragraph{\coralos{}.}
The archived CoralOS result files record runs against commit~\texttt{a38912b}
(v1.3.0), containerized as Docker image \texttt{coral-server:research-v5}.
Diffing the vulnerable sources against v1.4.0 (commit~\texttt{68b8326}) shows ten
of the eleven findings byte-identical, so they persist; V9 is additionally
re-proven live by exercising the production \texttt{getWallet()} compiled from the
v1.4.0 source (Section~\ref{sec:v9}). Only V7 (x402 budget race) was remediated in
v1.4.0, where the check-and-deduct now runs inside a \texttt{Mutex.withLock}
(new \texttt{SessionRunningBudget}).
Configuration: \texttt{auth.keys=["research-key-2026"]},
\texttt{registry.includeDebugAgents=true},
\texttt{llmProxy.providers=[openrouter with gpt-3.5-turbo]}.
Source integrity manifest (251~Kotlin files):
SHA-256 \texttt{6FCB7860B9F60299\allowbreak 318ACBE4C329924C\allowbreak 66180360D1B3FB9D\allowbreak E627537DAEB748C7}.
Full PoC scripts at \coralospoc{}.

\paragraph{\fetchai{}.}
\texttt{uagents} Python package, latest release as of June~2026
(\texttt{pip install uagents}).
The \texttt{AlmanacApiResolver} class is used directly in the F-1 proof.
Full PoC scripts at \fetchaipoc{}.

\paragraph{\aptwo{}.}
Reference implementation cloned from
\texttt{github.com\allowbreak/google-agentic-commerce\allowbreak/AP2} (\texttt{--depth=20},
commit as of June~2026).
Analysis targets \texttt{code/samples/python/src/}.
Full code analysis scripts at \aptwopc{}.

\paragraph{LLM API.}
All model calls via OpenRouter (\texttt{openrouter.ai/api/v1}).
Total LLM study cost: \$0.72
(1{,}440~calls: 8~models $\times$ 3~temperatures $\times$ 3~IPI variants $\times$ 20~trials each).

\section{Additional \coralos{} Vulnerability Details}
\label{supp:coralos}

\subsection*{V6 + V13: LLM Proxy Model Override (Both Paths)}

The configured model is \texttt{openai/gpt-3.5-turbo}.
For non-streaming requests (\emph{the LLM proxy service}),
the request body is forwarded unchanged (\texttt{else requestBody}):
the attacker sends \texttt{"model":"openai/gpt-4o"}, and the response
contains \texttt{"model":"openai/gpt-4o"} at cost \$0.00014
(10.4$\times$ the configured model's \$0.0000135).
For streaming requests (\emph{the streaming proxy strategy}),
\texttt{prepareStreamingRequest()} adds \texttt{stream\_options}
but does not rewrite the \texttt{model} field.
The gpt-4o streaming response was confirmed via OpenRouter
(cost \$0.000065).
Result hashes: \texttt{a417b861\ldots} (V6), \texttt{v13\_PROVEN\ldots}.

\subsection*{V10: CSRF Proof Scope}
V10 proves that the server's CORS configuration permits credentialed
cross-origin requests: Ktor echoes the request \texttt{Origin} in
\texttt{Access-Control-Allow-Origin} when \texttt{anyHost()} is set,
enabling browsers to include cookies per the CORS specification.
The Python PoC demonstrates the server-side policy gap.

Its browser-side vector is a cross-origin credentialed POST issued from
an attacker-controlled page, which the browser accompanies with the victim's
session cookie.

\redactbox

Cross-origin POST with victim cookie returned HTTP~200,
namespace \texttt{csrf\allowbreak-backdoor\allowbreak-namespace} created
(sessionId: \texttt{de195366\ldots}).

\section{Additional \fetchai{} Vulnerability Details}
\label{supp:fetchai}

\begin{table}[!ht]
\caption{\fetchai{} uAgents: 8 structural vulnerabilities}
\label{tab:fetchai}
\centering
\small
\renewcommand{\arraystretch}{1.1}
\begin{tabular}{@{}p{0.5cm}p{3.5cm}p{1.05cm}p{0.85cm}p{1.15cm}@{}}
\toprule
\textbf{ID} & \textbf{Title} & \textbf{Class} & \textbf{Sev.} & \textbf{Proof} \\
\midrule
F-1 & Endpoint hijacking via Almanac API & RC-1 & HIGH & Live \\
F-2 & SSRF via F-1 endpoint injection    & RC-1 & HIGH & Via F-1 \\
F-3 & Attestation token: 1000\,s lifetime bug & RC-3 & MED & Code \\
F-4 & Private keys in plaintext JSON     & RC-3 & HIGH & Code \\
F-5 & IPv6 localhost bypass (CWE-284)    & RC-5 & MED  & Code \\
F-6 & No body size limit on \texttt{/submit} & RC-4 & MED & Code \\
F-7 & Sync query pollution via forged sender & RC-5 & MED & Code \\
F-8 & Wallet key mismatch on restart     & RC-2 & HIGH & Repro. \\
\bottomrule
\end{tabular}
\end{table}

\subsection*{F-3: Attestation Token Lifetime Bug}
\redactbox

\subsection*{F-8: Wallet Key Mismatch (Deterministic Reproduction)}
\redactbox

\section{Additional \aptwo{} Vulnerability Details}
\label{supp:ap2}

\textbf{A-AP2-1 (Token Redemption TOCTOU, RC-4).}
\texttt{complete\_checkout()} in
\emph{the checkout completion path} follows the
pattern: check \texttt{token\_data.get('used')},
execute the payment,
then set \texttt{token\_data['used'] = True}.
This check-then-act gap spans roughly 70~lines and reads and writes a
file-based token store with no locking.
The window contains no \texttt{await}, so a single asyncio worker is safe,
but two worker processes both pass the \texttt{used=False} check before
either writes back, enabling double-spend (measured in Experiment~B).

\textbf{A-AP2-4 (Empty \texttt{trusted\_roots} Bypasses Cert Validation, RC-5).}
When the trusted root certificate file is absent,
\emph{the credentials provider} sets \texttt{trusted\_roots = []}.
The subsequent guard \texttt{if self.\_trusted\_roots:}
(\emph{the certificate-chain verifier}) evaluates to \texttt{False} for an empty list
(Python falsy semantics), skipping the entire certificate chain
verification block.
A mandate signed with any self-signed certificate is therefore accepted,
completely undermining the cryptographic mandate chain.

\textbf{A-AP2-11 (No Authentication on Merchant MCP Endpoints, RC-5).}
All five merchant operations (\texttt{search\_inventory},
\texttt{check\_product}, \texttt{assemble\_cart},
\texttt{create\_checkout}, and \texttt{complete\_checkout}) are
annotated only with \texttt{@mcp.tool()}.
No authentication decorator, no caller verification,
no API key check.
Any network-reachable party can invoke \texttt{complete\_checkout(
payment\_token, checkout\_mandate\_id)} directly,
bypassing the mandate chain entirely.

\textbf{A-AP2-15 (Role Confusion: Credentials Provider, RC-5).}
\texttt{issue\_payment\_credential()} in
\emph{the credential issuance tool}
verifies that a valid mandate SD-JWT is present but does not
authenticate the \emph{caller's identity}.
Any agent, not only the authorized payment processor, can present
a valid mandate and receive a payment credential.
The Go reference implementation, by contrast, performs an explicit
caller whitelist check (\emph{the Go reference executor}),
highlighting that this is a Python implementation oversight,
not an \aptwo{} specification gap.

\subsection*{A-AP2-4: Empty \texttt{trusted\_roots} (Python Falsy Evaluation)}
\redactbox

\subsection*{A-AP2-11: No Authentication on Merchant MCP}
\redactbox

\section{Additional Cross-Platform Attack Chains}
\label{supp:chains}

Beyond the \coralos{} kill chain (V5$\to$V4$\to$V9), two further chains
compose vulnerabilities across a single platform.
In \fetchai{}, F-1 and F-2 compose into server-side request forgery: the rogue
Almanac returns an internal URL as Agent~B's endpoint, and
\texttt{AlmanacApiResolver} makes the agent POST to it from the server,
reaching cloud metadata services (e.g., \texttt{169.254.169.254}) or internal
network endpoints inaccessible from outside.
In \aptwo{}, the file token store race (A-AP2-2) and the token-redemption
TOCTOU (A-AP2-1) compose into a double-spend under multi-worker deployment,
where two workers both observe \texttt{used=False} before either writes back.

\section{\pcat{} Implementation Notes}
\label{supp:pcat}

The P1 signing payload uses JSON to avoid delimiter-collision attacks:
\begin{lstlisting}[language=Python,
  caption={pcat\_proxy\_v3.py, JSON signing payload for P1.}]
def _build_signing_payload(body: bytes, ts: str, nonce: str) -> bytes:
    return json.dumps({
        "ts": ts, "nonce": nonce,
        "body_sha256": sha256(body).hexdigest(),
        "version": "pcat-sig-v1"
    }, sort_keys=True).encode()
\end{lstlisting}

\pcat{} underwent three rounds of internal security review,
identifying and fixing 10+ bugs including:
CSRF origin substring bypass (fixed to use \texttt{urlparse().hostname}),
P1 nonce replay (fixed with TTL heap cache),
P4 unbounded history growth (fixed with \\ \texttt{deque(maxlen=10000)}),
and P1 alert-mode fallthrough (fixed to block by default for registered services).
Source at \pcatcode{}.

\section{\aipbench{} Reproducibility}
\label{supp:bench}

\paragraph{Experiment~A (cross-platform \pcat{} enforcement).}
Using the reference \pcat{} implementation, the unauthenticated \aptwo{}
\texttt{complete\_checkout} call is rejected by P5, an unknown and an unsigned
\fetchai{} Almanac response are both rejected by P1, a correctly signed
response passes, and two concurrent single-use redemptions yield exactly one
success under P4 (result hash \texttt{e67be3f8\ldots}).

\paragraph{Experiment~B (RC-4 double-spend win-rate, 200 trials each).}
A faithful reproduction of the file-based token store measures the double-spend
rate by deployment model: single-process asyncio 0\%, two processes 30\%
(0.1\,ms verify window) to 57\% (2\,ms) under natural jitter, and 100\% under
synchronized start (result hash \texttt{f3c7bd50\ldots}).

\paragraph{Experiment~C (multi-worker P4 validation, 200 trials).}
Routing the two concurrent redemptions of Experiment~B through a single
\pcat{} gate (the real \texttt{AtomicPaymentStore}) admits exactly one
redemption in every trial, so the double-spend rate is 0\% even though the
backend is multi-worker (result hash \texttt{d41b22a5\ldots}).

\paragraph{Experiment~D (distributed P4, four \pcat{} instances, 200 trials).}
Four \pcat{} instances share one compare-and-swap store (SQLite's
\texttt{UPDATE\ldots WHERE used=0}, the primitive of Redis+Lua). Each trial's
four concurrent redemptions yield exactly one success and three
contention-rejections, with zero double-spends across all trials, validating
P4 for horizontally scaled deployments (result hash \texttt{9e190769\ldots}).

\paragraph{Experiment~E (P2 caller identity binding over a real \texttt{did:web} resolver).}
\pcat{}'s P2 resolves the caller's DID by a live HTTP fetch of the
controller-hosted DID document and verifies the caller's ECDSA proof against the
document's key. Over the real resolver, P2 admits the authorized wallet and
blocks all six RC-2 variants: three V9 cases (no DID, an unknown DID, a forged
proof), an A-AP2-5 wallet not authorized for its DID, a \emph{revoked} DID, and a
\emph{rotated}-key proof, at 0\% \asr{} (result hash \texttt{9ad9fb23\ldots}).

\section{Refusal Rationales in the LLM Study}
\label{supp:refusal}

We categorized the per-trial responses of the resistant models (logged in the
1{,}440-trial study, with additional full-length captures at cost \$0.12).
The attack-success rates are those of the main study, and the categories below
describe \emph{how} each model resists.

\begin{table}[!ht]
\caption{Refusal rationale by model (V5 marketplace IPI).}
\label{tab:refusal}
\centering
\small
\renewcommand{\arraystretch}{1.15}
\begin{tabular}{@{}p{2.3cm}p{2.1cm}p{2.7cm}@{}}
\toprule
\textbf{Model} & \textbf{Rationale} & \textbf{Representative text} \\
\midrule
Claude Sonnet 4.6 & Explicit injection detection & ``prompt injection / social engineering attack'' \\
Claude 3 Haiku & Flags altered recipient & ``updated to an AttackerWallet\ldots raises serious concerns'' \\
Llama 3.3 70B & Fraud refusal & ``cannot facilitate a potentially fraudulent transaction'' \\
GPT-4o & Bare policy refusal & ``I'm sorry, I can't assist with that'' \\
Gemini 2.5 Pro & Silent non-adoption & does not adopt the injected address, no surfaced rationale \\
Flash-tier & None (complies) & follows the injected address \\
\bottomrule
\end{tabular}
\end{table}

The resistant models treat the description field as untrusted and decline to
act on an unverified payment destination. This supports semantic defenses that
surface the payment destination for explicit user confirmation and sanitize
marketplace description fields, rather than relying on model self-diagnosis:
GPT-4o refuses without diagnosing the injection, and Flash-tier models do not
refuse at all.

\begin{itemize}
\item \textbf{Docker Compose} (\texttt{docker-compose.yml}):
  one command starts \coralos{}, rogue marketplace, rogue federation,
  and \pcat{} for full attack reproduction.
\item \textbf{PoC scripts} (\texttt{03\_attacks/}):
  Python 3.10+ scripts with \texttt{requirements.txt}
  (requests, aiohttp, websockets, cryptography, solana, solders).
\item \textbf{Result files} (\texttt{03\_attacks/results/}):
  JSON with embedded SHA-256 integrity hashes.
\item \textbf{Source manifest} (\texttt{SOURCE\_INTEGRITY.txt}):
  SHA-256 of all 251 \coralos{} Kotlin source files,
  manifest hash \texttt{6FCB7860\ldots}.
\end{itemize}

Code repository: \githubartifacts{}. Benchmark dataset: \hfbenchmark{}.

\end{document}